\documentclass[doublecolumn]{IEEEtran}
\usepackage{calc}
\usepackage[T1]{fontenc}

\usepackage[overload]{empheq}
\usepackage{cleveref}
\usepackage{multirow}
\usepackage{cite}
\usepackage{graphicx,subfigure}
\usepackage{psfrag}
\usepackage{amsmath,amssymb}
\usepackage{color}
\usepackage{tikz}
\usepackage{cleveref}
\usepackage{empheq}
\usepackage[numbers]{natbib}

\newenvironment{proof}{\begin{IEEEproof}}{\end{IEEEproof}}

\DeclareMathOperator*{\defeq}{\triangleq}

\newtheorem{proposition}{Proposition}

\newtheorem{lemma}{Lemma}

\newcommand{\bit}{\begin{itemize}}
\newcommand{\eit}{\end{itemize}}

\newcommand{\bc}{\begin{center}}
\newcommand{\ec}{\end{center}}

\newcommand{\ba}{\begin{array}}
\newcommand{\ea}{\end{array}}

\newcommand{\beq}{\begin{equation}}
\newcommand{\eeq}{\end{equation}}

\newcommand{\beqn}{\begin{equation*}}
\newcommand{\eeqn}{\end{equation*}}

\newcommand{\bean}{\begin{eqnarray*}}
\newcommand{\eean}{\end{eqnarray*}}
\newcommand{\bea}{\begin{eqnarray}}
\newcommand{\eea}{\end{eqnarray}}

\def\C{\mathbb{C}}
\def\E{\mathbb{E}}

\def\dv{\boldsymbol{d}}

\def\hv{\boldsymbol{h}}

\def\qv{\boldsymbol{q}}

\def\uv{\boldsymbol{u}}
\def\vv{\boldsymbol{v}}

\def\xv{\boldsymbol{x}}

\newtheorem{remark}{Remark}

\makeatletter
\def\blfootnote{\gdef\@thefnmark{}\@footnotetext}
\makeatother

\begin{document}
\sloppy

\title{Cloud-Edge Non-Orthogonal Transmission for Fog Networks with Delayed CSI at the Cloud}
\author{ Jingjing Zhang and Osvaldo Simeone  
\vspace{-1.1em}
	}  
\maketitle
\thispagestyle{empty}

\begin{abstract}
In a Fog Radio Access Network (F-RAN), the cloud processor (CP) collects channel state information (CSI)\blfootnote{The authors are with the Department of Informatics at King's College London, UK (emails: jingjing.1.zhang@kcl.ac.uk, osvaldo.simeone@kcl.ac.uk). The authors have received funding from the European Research Council (ERC) under the European Union's Horizon 2020 Research and Innovation Programme (Grant Agreement No. 725731).} from the edge nodes (ENs) over fronthaul links. As a result, the CSI at the cloud is generally affected by an error due to outdating. In this work, the problem of content delivery based on fronthaul transmission and edge caching is studied from an information-theoretic perspective in the high signal-to-noise ratio (SNR) regime. For the set-up under study, under the assumption of perfect CSI, prior work has shown the (approximate or exact) optimality of a scheme in which the ENs transmit information received from the cloud and cached contents over orthogonal resources. In this work, it is demonstrated that a non-orthogonal transmission scheme is able to substantially improve the latency performance in the presence of imperfect CSI at the cloud.
\end{abstract}


%
%


\section{Introduction}

Consider the scenario shown in Fig.~\ref{fig:model}, in which a wireless cellular system delivers contents to a number of users by means of a centralized CP with full access to the library content, fronthaul links, and ENs endowed with caching capabilities. This set-up, referred to as a F-RAN, is motivated by current trends in the evolution of wireless cellular systems.

The F-RAN model has been recently studied from an information-theoretic perspective \cite{ZS:18}, with special cases excluding CP and fronthaul links investigated in \cite{HND:16,NMA:17}. These works consider the high-SNR regime in order to focus on the impact of interference. Reference \cite{ZS:18} proves that, in the presence of full CSI at both CP and ENs, it is approximately optimal to have the ENs transmit information received from the cloud and cached contents over orthogonal resources using time division multiplexing. 

In practice, as seen in Fig.~\ref{fig:model}, the cloud processor collects CSI from the ENs over fronthaul links. The ENs, instead, can directly receive CSI from the users via feedback under a Frequency Division Multiplexing (FDD) operation. As a result, the CSI at the CP is generally affected by an additional error due to outdating associated to fronthaul transmission latency.

In \cite{CJHR:16}, a novel approach is proposed that is shown to improve the high-SNR performance of multi-antenna systems in the presence of imperfect CSI. The scheme is based on rate splitting and superposition coding. Accordingly, a message of interest for multiple receivers is transmitted on the same radio resources as private messages intended for individual users in a non-orthogonal fashion with a proper power allocation.

In this work, it is demonstrated that for an F-RAN with imperfect CSI at the CP, a non-orthogonal transmission scheme is able to substantially improve the latency performance in the presence of imperfect CSI at the cloud in the high-SNR regime. Unlike \cite{CJHR:16}, the signal being multiplexed non-orthogonally are the information received from the cloud on the fronthaul links and locally encoded functions of the cached contents. A similar superposition scheme, known as hybrid fronthauling, was first proposed in \cite{PSS:16}. There, assuming full CSI, it was shown via numerical results that the scheme offers performance advantages in the finite-SNR regime, particularly for lower cache capacities.

\textbf{Notation:}
For any integer $K$, we define the set $[K]\defeq\{1,2,\cdots,K\}$. We also define the notation $\{f_n\}_{n=1}^{N}\defeq\{f_1,\cdots,f_N\}$. For a set $\mathcal{A}$, $|\mathcal{A}|$ represents the cardinality. We use the symbol $\doteq$ to denote an exponential equality in the sense that we write $f(P)\doteq P^\alpha$ if the limit $\lim_{P\rightarrow\infty} \log f(P)/\log P=\alpha$ holds. 

\begin{figure}[t!] 
  \centering
\includegraphics[width=0.7\columnwidth,height=4.5cm]{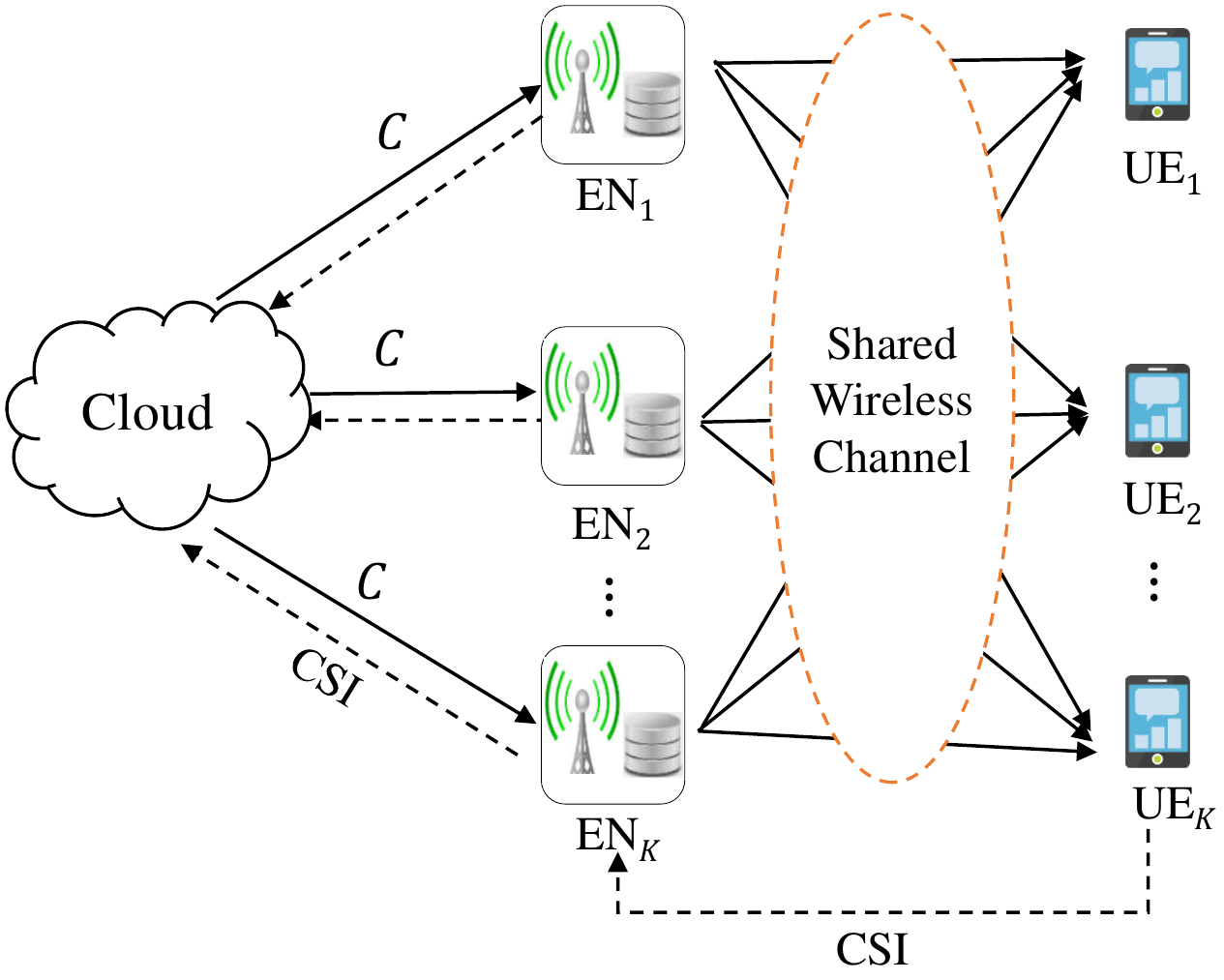}
\vspace{-0.8em}
\caption{Cloud and cache-based F-RAN system with solid lines representing downlink communication for content delivery and dashed lines indicating uplink CSI feedback.}
\label{fig:model}
\vspace{-1.5em}
\end{figure}

\section{System Model}

\subsection{F-RAN Model}

As shown in Fig.~\ref{fig:model}, we consider a F-RAN model in which $K$ ENs serve $K$ users through a shared wireless channel, with all nodes having a single antenna. We study downlink content delivery from a library of $N$ files $\{W_n\}_{n=1}^{N}$, each with size $L$ bits. Each EN is equipped with a cache, which can store $\mu NL$ bits from the library during the offline caching phase, with $\mu \in [0,1]$ being the fractional cache capacity. A CP connects to each EN via a dedicated fronthaul link with capacity $C$ bits per symbol, and it has access to the whole library. A symbol refers to a channel use of the wireless channel. 

In each dedicated time-slot, the signal received at each symbol for each user $k$ is given as
\begin{align} \label{transmit}
y_k=\hv^T_k\xv+n_k= \sum_{i=1}^{K} h_{k}^{i} x_i + n_k, 
\end{align}
where we have defined the vectors $\hv_{k}=[h_{k}^{1},\cdots,h_{k}^{K}]^T$ and $\xv=[x_1,\cdots,x_K]^T$ with $x_i$ being the transmitted signal of EN $i$; $h_{k}^{i}$ is the complex channel coefficient between EN $i$ and user $k$; and $n_k$ is a zero-mean complex Gaussian noise with normalized unit power. Channels are independent, drawn from a continuous distribution, and constant within each time-slot. We impose the power constraint $\E[|x_i|^2] \leq P$ for each EN $i$.

As illustrated in Fig.~\ref{fig:model}, we assume an FDD operation in which the users perform channel estimation based on downlink training, and then feed back the estimated CSI to the ENs over an uplink control channel. The CSI is finally communicated from the ENs to the cloud over fronthaul links. To study the problem in its simplest instantiation, we assume here that the feedback channel from users to ENs is ideal, so that each EN has full information about all channel gains $\textbf{H}=\{h_{k}^{i}\}$, with $i, k\in[K]$. In contrast, owing to fronthaul transmission delays and processing, the CSI available at the cloud is assumed to be delayed. The distortion of the outdated CSI $\{\hat{h}_{k}^{i}\}$ available at the cloud, i.e., the innovation $h_{k}^{i}-\hat{h}_{k}^{i}$, is characterized by its power $\E[|h_{k}^{i}-\hat{h}_{k}^{i}|^2]$. Following the standard model considered in, e.g., \cite{DJ:14,GJ:12o}, we assume that this mean squared error scales as $P^{-\alpha}$ with respect to the SNR $P$, for some $\alpha \geq 0$. More formally, we impose the exponential equality
\begin{align} \label{def:scaling}
\E[|h_{k}^{i}-\hat{h}_{k}^{i}|^2]\doteq P^{-\alpha}.
\end{align}
In the high-SNR regime, the case $\alpha=0$, which corresponds to finite-precision CSI, is equivalent to having no CSIT; while, at the other extreme, the case $\alpha=1$ yields a negligible CSI error. We study the general case in which the cloud has imperfect CSI in the sense of \eqref{def:scaling} with any arbitrary value $\alpha\in[0,1]$.

\subsection{Caching and Delivery Policies}

The communication protocol consists of caching and delivery phases.

1) \emph{Caching phase:} In the caching phase, the cache of each EN $i$ is proactively filled with information from the content library. More precisely, each file $W_n$ is mapped to a cached content $V_{n}^i$ by an arbitrary function $f_{n}^i$ as $V_{n}^i=f_{n}^i(W_n)$. To satisfy the capacity limitation of the cache, we have the inequality $\log_2 |V_{n}^i|\leq \mu L$ bits, where $|V_{n}^i|$ represents the alphabet of variable $V_{n}^i$. The overall cache content of EN $i$ is hence given as $V^i=\{V_{n}^i\}_{n=1}^{N}$. 

2) \emph{Delivery phase:} In the delivery phase, for any demand vector $\dv=[d_1,d_2,\cdots,d_K]$, where $W_{d_k}$ is the file requested by user $k$, the delivery of the $K$ files is completed via fronthaul and edge transmission. Fronthaul transmission occurs first with a  duration of $T_F$ symbols. On each $i$th fronthaul link, the CP sends a message $U^{i}$ about the requested files to EN $i\in[K]$. This message is obtained as a function of the available channel estimates $\hat{\textbf{H}}=\{\hat{h}_{k}^{i}\}_{i,k\in[K]}$, of the requested files, and of information about the cached files, as $U^i=g^i_{f}\big(\dv,\hat{\textbf{H}}, \{V^{i}\}_{i=1}^{K}\big)$. By the fronthaul capacity constant, we have the condition $\log_2 |U^i|\leq T_F C$ bits. Fronthaul transfer is followed by edge transmission, whereby each EN $i$ sends $T_E$ symbols obtained as the function $x_i=g^i_{e}\big(\dv,\textbf{H}, U^i,V^{i}\big)$ on channel \eqref{transmit}.


\subsection{Performance Metric: NDT}

A sequence of policies defined by functions $\big\{\{f_{n}^i\}, g_{f}^i,g^i_{e}\big\}$ is feasible if each user $k$ is able to decode the desired file $W_{d_k}$ with negligible probability of error when $L\rightarrow \infty$. For any feasible policy, we are interested in the delivery latency performance in the high-SNR regime. As in \cite{STS:17}, we parameterize the fronthaul capacity as $C= r \log(P)$, where $r$ is the fronthaul rate. Furthermore, we normalize the delivery latency $T_F+T_E$, where $T_F$ and $T_E$ are the corresponding fronthaul and edge latencies, by the term $L/\log(P)$. This represents the downlink latency of an ideal system where each user is served without interference at the high-SNR capacity $\log P$ bit/symbol via the wireless channels \cite{STS:17}. As a result, the fronthaul and edge NDTs are defined as 
\begin{align}
\delta_F = \lim_{P\rightarrow\infty} \lim_{L\rightarrow\infty} \frac{\E[T_F]}{L/\log(P)} ~ \text{and} ~\delta_E =\lim_{P\rightarrow\infty} \lim_{L\rightarrow\infty}  \frac{\E[T_E]}{L/\log(P)}. \notag
\end{align}
The overall NDT achieved by a sequence of feasible policies is given as 
\begin{align} \label{def:deltaach}
\delta = \delta_E+\delta_E.
\end{align}
For given parameters $(\mu,r,\alpha)$, the minimum NDT across all feasible policies is denoted as $\delta^*(\mu,r,\alpha)$.

\section{Orthogonal cloud-edge delivery} \label{sec:OMA}

In this section, we study the NDT performance of various policies based on state-of-the-art caching and delivery strategies, as described in \cite{STS:17,TS:16,HND:16,NMA:17}. According to these approaches, the ENs transmit information obtained from the caches and from the cloud on the fronthaul links in orthogonal time-slots. We refer to this class of techniques as performing \emph{orthogonal cloud-edge delivery}. As shown in \cite{STS:17}, the mentioned orthogonal strategies yield the minimum NDT in the presence of perfect CSI at the cloud, i.e., with $\alpha=1$, when $K=2$, and are more generally optimal within a multiplicative factor of two. We start with edge-based and cloud-based policies, which are then orthogonally multiplexed by means of time-sharing. It is noted that cloud-based soft-transfer fronthauling, proposed in \cite{STS:17}, is generalized here to account for imperfect CSI at the cloud (Lemma~\ref{lem:cs}). 

\subsection{Edge-based Policies} \label{subsec:eb}

We first consider caching polices based on only edge resources, for which the NDT performance does not depend on the CSI quality $\alpha$ at the cloud. Note that here we have zero fronthaul NDT, and hence the NDT \eqref{def:deltaach} is given as $\delta=\delta_E$.

\emph{1) Edge-based ZF-beamforming}. When the fractional cache capacity is $\mu=1$, full cooperation at the ENs is possible given that all ENs store the entire library. As a result, the ideal NDT of 1 can be achieved by means of ZF beamforming \cite[Lemma 2]{STS:17}, i.e., we have the achievable NDT
\begin{align}  \label{ndt:ez}
\delta_{EZ} =1.
\end{align}


\emph{2) Edge-based interference alignment (IA)}. Consider now the case with cache capacity $\mu=1/K$. This is the minimum cache size allowing for delivery based only on cached contents. By caching disjoint fractions of all files and using X channel IA for delivery \cite[Lemma 3]{STS:17}, the following NDT is achievable
\begin{align}  \label{ndt:ei} 
\delta_{EI}=2-\frac{1}{K}.
\end{align}


\subsection{Cloud-based Policies} \label{sebsec:cb}

We now focus on cloud-based policies assuming that the ENs have zero cache capacity, i.e., $\mu=0$. As discussed in \cite{STS:17}, we can distinguish hard and soft-transfer fronthauling approaches. The former techniques send uncoded fractions of files via the fronthaul links. In contrast, the latter method implements ZF beamforming at the cloud, and sends quantized precoded signals to the ENs. Unlike hard-transfer fronthauling, soft-transfer fronthauling relies on perfect CSI at the cloud in order to perform ZF beamforming. Here, we generalize the analysis of the corresponding achievable NDT to the case of imperfect CSI.

\emph{3) Cloud-based hard-transfer fronthauling}. With zero cache capacity, i.e., with $\mu=0$, the NDT 
\begin{align}  \label{ndt:ch}
\delta_{CH} =\min\Big\{1+\frac{K}{r},2-\frac{1}{K}+\frac{1}{r}\Big\}
\end{align}
is achievable by using one of the following hard-transfer fronthauling approaches, which attain the two NDTs in \eqref{ndt:ch}. In the first scheme, the cloud sends all the requested files to each EN via the respective fronthaul link, and the ENs carry out cooperative ZF beamforming. In the second scheme, the requested files are split into disjoint fragments and sent to the ENs, which performs X channel IA (see \cite[Proposition 2]{STS:17}). 

\emph{4) Achievable NDT with cloud-based soft-transfer fronthauling}. With soft-transfer fronthauling, we have the achievable NDT described in the following lemma. 
\begin{lemma} \label{lem:cs}
In an $K\times K$ F-RAN with CSI quality $\alpha \in [0,1]$ at the cloud and cache capacity $\mu=0$, the NDT 
\begin{align}  \label{ndt:cs}
\delta_{CS}=\frac{1}{r}+\frac{1}{\alpha}
\end{align}
is achievable by means of soft-transfer fronthauling. 
\end{lemma}

\begin{proof}
For any request vector $\dv$, the cloud precodes the $K$ requested files producing the $K \times 1$ vector
\begin{align} \label{precode}
\bar{\xv}_F=\sum_{k=1}^{K} \vv_{d_k} s_{d_k},
\end{align}
where $s_{d_k}$ is a symbol of the codeword encoding file $W_{d_k}$, and $\vv_{d_k} \in \C^{K\times 1}$ is the corresponding beamforming vector. The symbols $\{s_{d_k}\}_{k=1}^{K}$ are taken from a Gaussian codebook of equal rate $R_F$ bits/symbol and equal power. The power of the symbols $s_{d_k}$ is set to be exponentially equal to $P$, i.e., $\E[|s_{d_k}|^2]\doteq P$. The unit-norm precoding vector $\vv_{d_k}$ is selected to be orthogonal to all the estimates of the vectors $\{\hv_{k'}\}_{k'\in[K], k'\neq k}$, i.e., we have the equalities $\hat{\hv}_{k'}^T\vv_{d_k}=0$ for all $k'\neq k$. The cloud then quantizes each element $\bar{x}_{Fi}$ of the vector signal $\bar{\xv}_F=[\bar{x}_{F1},\cdots,\bar{x}_{FK}]^T$ as
\setlength{\abovedisplayskip}{-0.1pt} 
\begin{align} \label{quantization}
x_{Fi}=\bar{x}_{Fi}+q_i,
\end{align}  
for $i\in[K]$, where $q_i$ is the quantization noise, which is modeled as a Gaussian variable $q_i \sim \mathcal{CN} (0,\sigma^2)$ with variance $\sigma^2$, which is independent across index $i$. By standard rate distortion arguments, the variance $\sigma^2$ is related to the number $B$ of bits for each of the $L/R_F$ samples $x_{Fi}$ by the equalities $B=I(\bar{x}_{Fi};x_{Fi})=\log(1+\E[|\bar{x}_{Fi}|^2]/\sigma^2)$, and hence we have the equality $\sigma^2=\E[|\bar{x}_{Fi}|^2]/(2^{B}-1)$. We choose $B=\alpha \log P$, so that we have the exponential equality $\sigma^2 \doteq P/ P^\alpha = P^{1-\alpha}$. Note that the power constraint $\E[|x_{Fi}|^2]\doteq P$ is satisfied. The cloud sends the quantized signal $x_{Fi}$ to EN $i$ for $i\in[K]$ at the fronthaul rate $C=r\log P$ on the fronthaul links. As a result, the fronthaul latency is given as $T_F=B(L/R_F)/C=\alpha L/(R_F r)$. 


Each EN $i$ then forwards the quantized signal $x_{Fi}$, i.e., we set $x_i=x_{Fi}$ in \eqref{transmit}, and hence each user $k$ receives the signal
\begin{subequations}
\begin{align} \label{rec}
\! y_k&\!=\hv_{k}^T \xv_F+n_k= \hv_{k}^T \Bigg(\sum_{k=1}^{K} \vv_{d_k} s_{d_k} +\qv\Bigg)+n_k  \\
 \! &\!= \hv_{k}^T \vv_{d_k} s_{d_k} +\hv_{k}^T \Bigg(\sum_{k'=1, k'\neq k}^{K} \vv_{d_{k'}} s_{d_{k'}} +\qv \Bigg)\!+\!n_k\\
 \! &\! \stackrel{(a)}{=}\hv_{k}^T \vv_{d_k} s_{d_k}\! + \underbrace{\tilde{\hv}_{k}^T\sum_{k'=1, k'\neq k}^{K} (\vv_{d_{k'}} s_{d_{k'}} )}_{\defeq z_k} +\hv^T_{k} \qv \! + \! n_k, 
\end{align}
\end{subequations}
where we have defined the vectors $\qv=[q_1, \cdots,q_K]^T$ and $\tilde{\hv}_{k}=\hv_{k} -\hat{\hv}_{k}$; and equality (a) holds due to the conditions $\hv_{k}^T \vv_{d_{k'}}= (\hat{\hv}_{k}^T+\tilde{\hv}_{k}^T) \vv_{d_{k'}}=\tilde{\hv}_{k}^T\vv_{d_{k'}}$. Given the CSI error scaling \eqref{def:scaling}, the power of the interference $z_{k}$ satisfies the exponential equality $\E[|z_{k}|^2]\doteq P^{1-\alpha}$. 
Furthermore, the power of the effective noise, namely quantization plus Gaussian noise, is given as $\E[|\hv_{k}^T \qv|^2]+\E[|n_k|^2] \doteq P^{1-\alpha}$. It follows that the file $W_{d_k}$, encoded by $s_{d_k}$, can be decoded reliably in the high-SNR regime with rate $R_F=\log (P/P^{1-\alpha})=\alpha \log P$ by treating interference as noise. Finally, we have the fronthaul NDT $\delta_F=1/r$ and the duration of edge transmission is given by $T_E=L/(\alpha \log P)$, yielding the edge NDT $\delta_E=1/\alpha$. This completes the proof. 
\end{proof}

\subsection{Cloud and Edge-based Policies}
By means of time-sharing of the cloud- and edge-based solutions described above, cloud-edge orthogonal delivery achieves the following NDT. 

\begin{proposition} \label{ndtp}
\emph{Achievable NDT with cloud-edge orthogonal transmission}. In an $K\times K$ F-RAN with CSI quality $\alpha \in [0,1]$ at the cloud and cache capacity $\mu>0$, the following NDT is achievable by cloud-edge orthogonal delivery
\begin{align}  \label{ndt:sp}
\delta_{O}(\mu, r,\alpha) =\min\{(\delta_{EZ}-\delta_C)\mu+\delta_C, \delta'_{O}(\mu, r,\alpha)\}
\end{align}
where we have defined $\delta_C=\min\{\delta_{CH}, \delta_{CS}\}$ and 
\begin{align}
\delta'_{O}(\mu, r,\alpha)=\left\{
\begin{array}{ll}
       2-\mu, &\text{if}~ \mu \geq \frac{1}{K} \\ 
       (\delta_{EI}-\delta_C)K\mu+\delta_C, & \text{if}~ \mu \leq \frac{1}{K}.
\end{array} 
\right.
\end{align}
\end{proposition}

\begin{proof}
The result follows by the standard time sharing \cite[Lemma 1]{STS:17}.
\end{proof}

\section{Non-orthogonal Cloud-edge Delivery}

In this section, we introduce and analyze the proposed non-orthogonal cloud-edge delivery scheme. Unlike the orthogonal strategies studied in the previous section, cloud-based delivery, which is based on soft-transfer fronthauling, and edge-based delivery, which leverages the cached contents, are performed simultaneously at the ENs. As illustrated in Fig.~\ref{fig:scheme}, the technique is specifically based on the superposition at the ENs of the quantized fronthaul signals precoded at the cloud and of the signals encoding cached information,  as well as on successive interference cancellation (SIC) at the users. 

\begin{proposition} \label{pro:superposition}
\emph{Achievable NDT with cloud-edge non-orthogonal transmission.} 
In an $K\times K$ F-RAN with CSI quality $\alpha \in [0,1]$ at the cloud and cache capacity $\mu>0$, the NDT 
\begin{align}  \label{ndt:sp}
\delta_{NO}(\mu, r,\alpha) =\left\{
\begin{array}{ll}
       1+\frac{1-\mu}{r}, &\mu \geq 1-\alpha \\ 
       (1-\mu)(\frac{1}{\alpha}+\frac{1}{r}), & \mu \leq 1-\alpha,
\end{array} 
\right.
\end{align}
is achievable by means of non-orthogonal cloud-edge delivery. 
\end{proposition}

\begin{figure}[t!] 
  \centering
\includegraphics[width=1.03\columnwidth]{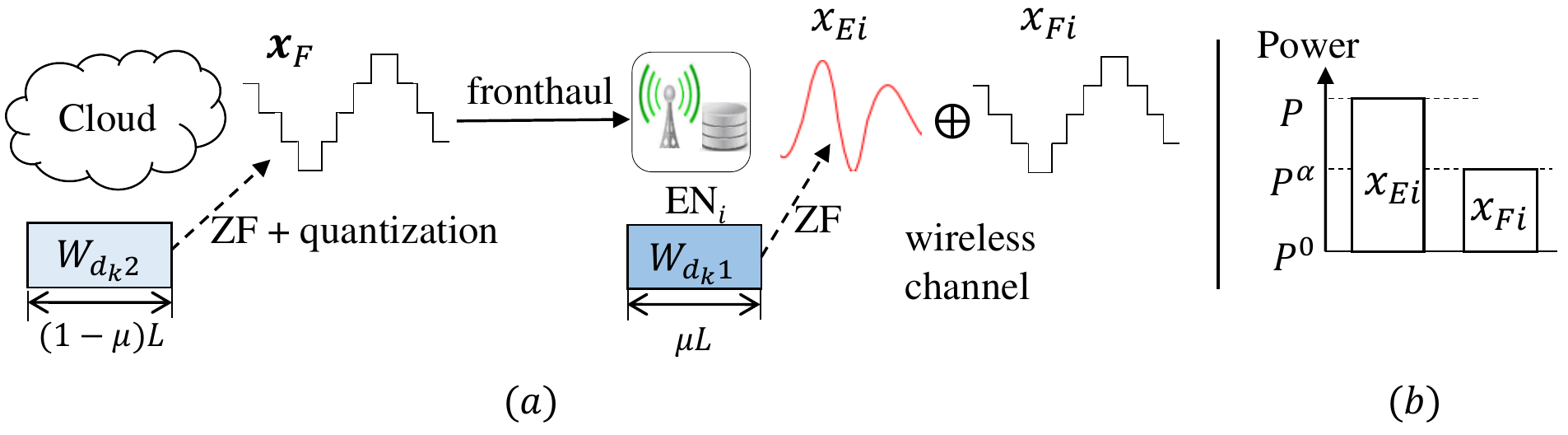}
\vspace{-1.8em}
\caption{Illustration of non-orthogonal cloud-edge delivery: (a) caching and delivery; and (b) powers of the signals $x_{Ei}$ and $x_{Fi}$ at each EN $i$.} 
\label{fig:scheme}
\vspace{-1.5em}
\end{figure}

A sketch of the proof is as follows and a full proof can be found in the Appendix. As illustrated in Fig.~\ref{fig:scheme}, in the caching phase, a fraction $\mu$ of each file is stored at all ENs. In the delivery phase, the cloud precodes the uncached $(1-\mu)$-fraction of the requested files using the cloud-based soft-transfer fronthauling scheme detailed in the proof of Lemma~\ref{lem:cs}. The quantized signals are then sent to the ENs via the fronthaul links. The ENs perform cooperative ZF precoding for the fraction $\mu$ of the cached contents using edge-based ZF beamforming as described in Section~\ref{subsec:eb}. The resulting fronthaul and locally encoded signals are summed and sent to the users, with the former being transmitted with a lower power than the latter. Each user decodes the two signals by using SIC: the locally precoded signal, which has a higher power, is decoded first by treating the fronthaul precoded signal as noise. This signal is then removed from the received signals, and, finally, the fronthaul precoded signal is decoded by the user.

\begin{remark}
In a $K\times K$ F-RAN with $r\geq 1$ and $\mu \geq 1-\alpha$, the NDT \eqref{ndt:sp} coincides with the NDT derived in \cite{STS:17} for the case of perfect CSI. This NDT is proved in \cite{TS:16} to be optimal for $K=2$ and to be generally within a multiplicative gap of 2 from the minimum NDT under perfect CSI. This suggests that imperfect CSI at the cloud may not cause any performance degradation as long as $\mu$ and $r$ are sufficiently large if non-orthogonal transmission is used. 
\end{remark}

%

Finally, combining orthogonal and non-orthogonal cloud-edge approaches, an improved achievable NDT can be obtained by means of time-sharing. 

\begin{proposition} \label{pro:improve}
\emph{(Achievable NDT)}.  In an $K\times K$ F-RAN with CSI quality $\alpha \in [0,1]$ at the cloud and cache capacity $\mu>0$, the following NDT is achievable 
\begin{align}  \label{ndt}
\delta(\mu, r,\alpha) =\text{l.c.e.} \big(\delta_{O}(\mu, r,\alpha), \delta_{NO}(\mu, r,\alpha)\big),
\end{align}
where the lower convex envelope (l.c.e.)\footnote{The l.c.e., is the supremum of all convex functions that lie under the given function.} is evaluated with respect to $\mu$.
\end{proposition}


%

\begin{figure}[t!] 
  \centering
\includegraphics[width=0.68\columnwidth]{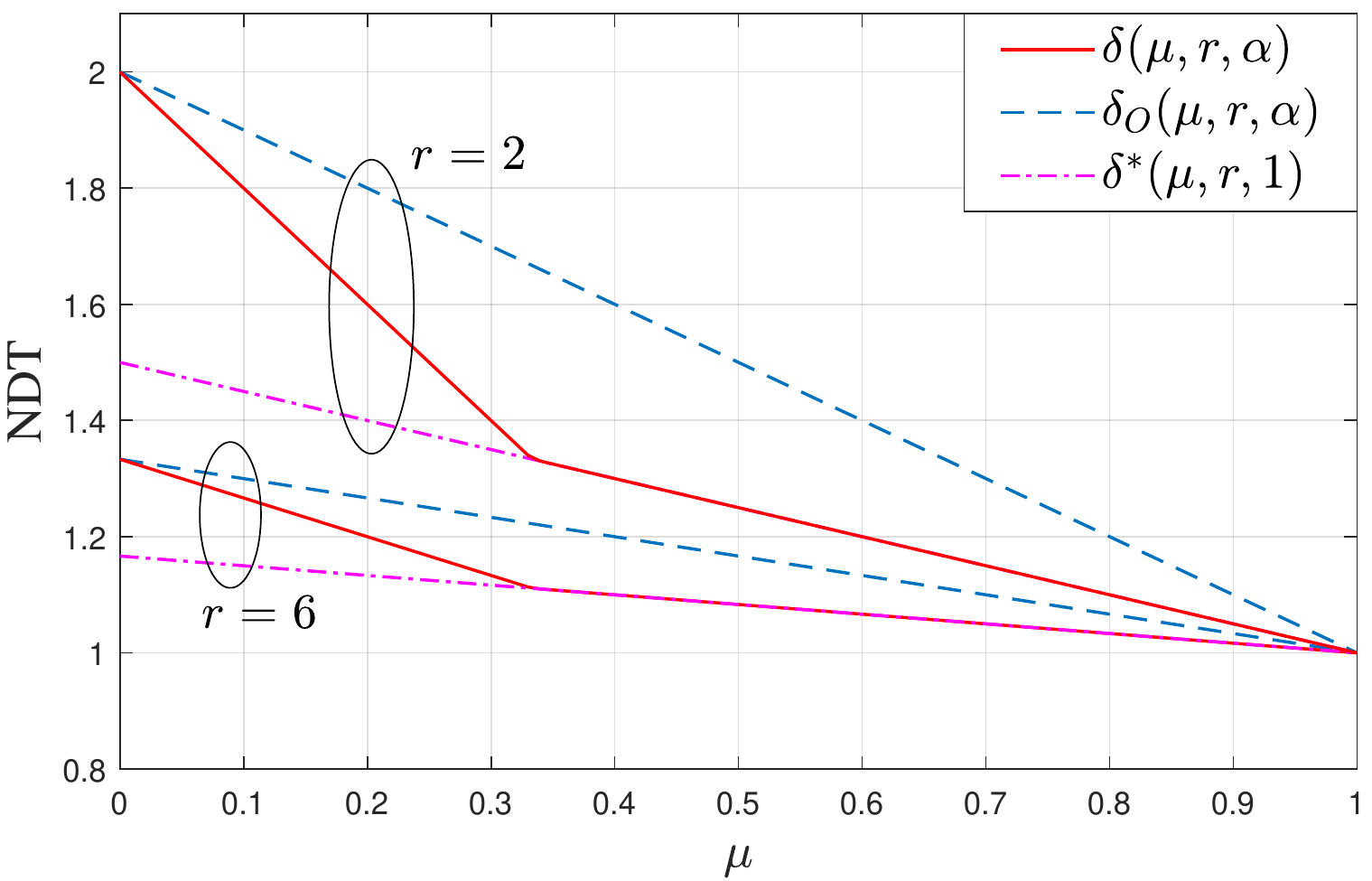}
\vspace{-0.8em}
\caption{NDT $\delta(\mu, r,\alpha)$ in \eqref{ndt} with non-orthogonal delivery and $\delta_{O}(\mu, r,\alpha)$ in \eqref{ndt:sp} with orthogonal delivery, and minimum NDT $\delta^*(\mu, r,1)$ with $\alpha=1$, versus $\mu$ for the $2\times 2$ F-RAN with different values of $r$ and $\alpha=2/3$.}
\label{fig:gain}
\vspace{-0.6em}
\end{figure}

\begin{figure}[t!] 
  \centering
\includegraphics[width=0.68\columnwidth]{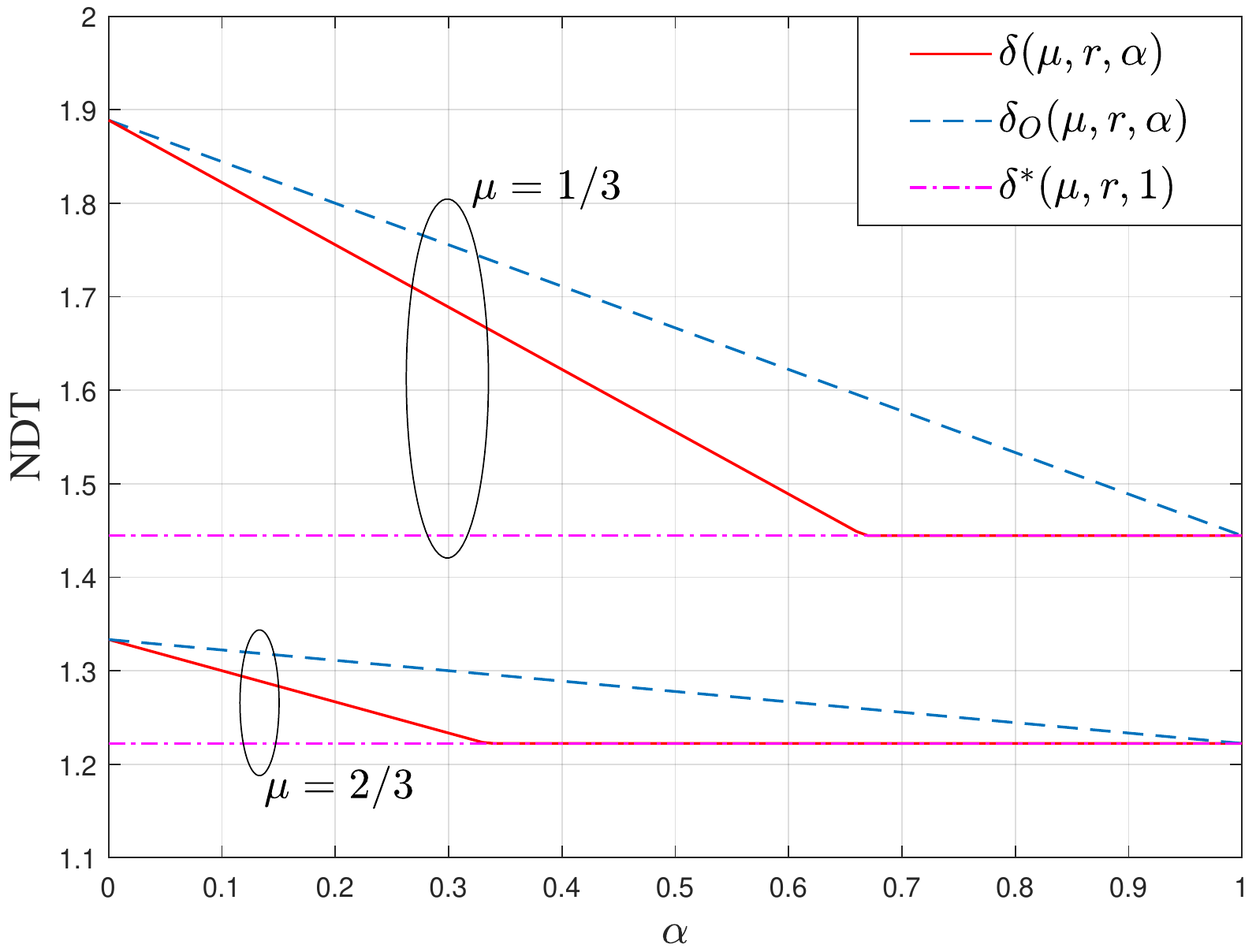}
\vspace{-0.8em}
\caption{NDT $\delta(\mu, r,\alpha)$ in \eqref{ndt} with non-orthogonal delivery and $\delta_{O}(\mu, r,\alpha)$ in \eqref{ndt:sp} with orthogonal delivery versus $\alpha$, and minimum NDT $\delta^*(\mu, r,1)$ with $\alpha=1$, for the $2\times 2$ F-RAN with $r=1.5$ and different values of $\mu$.}
\label{fig:gaina}
\vspace{-1.7em}
\end{figure}


\section{Numerical Examples} 



In this section, we compare the NDT $\delta_{O}(\mu, r,\alpha)$ in \eqref{ndt:sp} with conventional orthogonal delivery, the NDT $\delta(\mu, r,\alpha)$ in \eqref{ndt} with non-orthogonal delivery, and the minimum NDT $\delta^*(\mu,r,1)$ derived in \cite{TS:16} for the case of perfect CSI ($\alpha=1$) at the cloud. Throughout, we set $K=2$.
 
Fig.~\ref{fig:gain} plots the mentioned NDTs with $\alpha=2/3$ for both cases $r=2$ and $r=6$. It is observed that latency gains can be reaped via non-orthogonal delivery for the whole range of values $\mu\in (0,1)$. The largest gains are obtained in the intermediate regime where both cloud and edge-caching contributions are similarly relevant. Furthermore, as observed in Remark 1, when $\mu \geq 1-\alpha=1/3$, the achievable NDT $\delta(\mu, r,\alpha)$ coincides with the minimum NDT $\delta^*(\mu,r,1)$ obtained with perfect CSI. Note that this result can be achieved with orthogonal delivery only when $\mu=1$. 

Fig.~\ref{fig:gaina} plots the NDTs versus the CSI quality $\alpha$. As $\alpha$ increases, non-orthogonal delivery benefits more than orthogonal delivery from the improved CSI, and it obtains the latency savings for any value of $\alpha\in(0,1)$. Moreover, when $\alpha \geq 1-\mu$, as observed in Remark 1, non-orthogonal delivery can obtain the optimal performance under perfect CSI. 
\vspace{-0.8em}

\section{Conclusions}
In this paper, we studied the problem of content delivery in F-RAN in the presence of heterogeneous CSI availability between edge and cloud. 
A non-orthogonal transmission scheme superimposes signals produced at the cloud and at the edge based on cached contents was shown to reduce delivery latency as compared to conventional orthogonal methods. The approach can obtain optimal full-CSI performance for sufficiently large cache and fronthaul capacities. 

\section{appendix: proof of Proposition \ref{pro:superposition}}
In the caching phase, each file is split into two disjoint subfiles, i.e., $W_n=\{W_{n1}, W_{n2}\}$, where $W_{n1}$ has size $\mu L$ bits. The first subfiles $\{W_{n1}\}_{n=1}^{N}$ are placed in each EN's cache. In the delivery phase, consider any demand vector $\dv$. The cloud precodes the uncached subfiles $\{W_{d_k 2}\}_{k=1}^{K}$ producing signal $\bar{\xv}_F=\sum_{k=1}^{K} \vv_{d_k} s_{d_k}$. Unlike the signal in \eqref{precode}, here each symbol $s_{d_k}$ only encodes subfile $W_{d_k 2}$ instead of the whole file. Moreover, the power of the symbols $s_{d_k}$ is set to satisfy the exponential equality $\E[|s_{d_k}|^2]\doteq P^{\alpha}$. Similar to \eqref{quantization}, upon fronthaul quantization, the signal $x_{Fi}=\bar{x}_{Fi}+q_i$ for each EN $i$ is produced, where $q_i \sim \mathcal{CN} (0,\sigma^2)$ is the quantization noise. We again set $B=\alpha \log P$ as in Lemma~\ref{lem:cs}. As a result, we have the exponential equality $\sigma^2\doteq P^{\alpha}/P^{\alpha}=1$. Hence, the fronthaul latency is given as $T_F=B(L(1-\mu)/R_F)/C=BL(1-\mu)/(R_F C)$, yielding the fronthaul NDT $\delta_F=\alpha \log P (1-\alpha)/(R_F r)$.


The cached subfiles $\{W_{d_k1}\}_{k=1}^{K}$ available at all ENs are precoded cooperatively using edge-based ZF beamforming. To elaborate, the $K \times 1$ precoded signal produced across the ENs is given as 
\begin{align} \label{precode:edge}
\xv_E=\sum_{k=1}^{K} \uv_{d_k} c_{d_k},
\end{align}
where symbol $c_{d_k}$ denotes the codeword encoding the cached subfile $W_{d_k1}$, with rate $R_E$ bits/symbol; and the corresponding precoder $\uv_{d_k} \in \C^{K\times 1}$ is designed to be orthogonal to all the channels $\{\hv_{k'}\}_{k'\in[K], k'\neq k}$, i.e., $\hv_{k'}^T\uv_{d_k}=0$. Recall that this is feasible since the ENs have perfect CSI. The edge encoded signals $c_{d_k}$ have power $\E[|c_{d_k}|^2]\doteq P$.

The ENs superimpose the edge precoded signal $\xv_E$ with the cloud precoded and quantized signal $\xv_F$, yielding the signal $\xv=\xv_E+\xv_F$.
 Note that the power constraint $\E[|x_i|^2]\doteq P$ is satisfied. 
As a result, the received signal at user $k$ is given as 
\begin{subequations}
\begin{align} 
\!\!\! y_k&\!=\hv_k^T\xv+n_k  \\
         & =\hv_k^T(\xv_E+\xv_F)+n_k \\
 \!\! \!  &\!=\hv_k^T\Bigg(\sum_{k=1}^{K} \uv_{d_k} c_{d_k}+ \sum_{k=1}^{K} \vv_{d_k} s_{d_k}+\qv\Bigg) +n_k \\
	\!\!\! &\! \stackrel{(a)}{=} \underbrace{\hv_k^T  \uv_{d_k} c_{d_k}}_{\defeq \tilde{x}_{E,k}}+ \underbrace{\hv_{k}^T \vv_{d_k} s_{d_k}}_{\defeq \tilde{x}_{F,k}}  +  \underbrace{\tilde{\hv}_{k}^T \sum_{k'=1, k'\neq k}^{K}  (\vv_{d_{k'}} s_{d_{k'}} )}_{\defeq z_k}  \notag \\
	&~~~ +\hv^T_{k} \qv + n_k,  \label{eq:receive}
\end{align}
\end{subequations}
where equality (a) holds due to the conditions $\hv_{k}^T \uv_{d_{k'}}=0$ and $\hv_{k}^T \vv_{d_{k'}}= (\hat{\hv}_{k}^T+\tilde{\hv}_{k}^T) \vv_{d_{k'}}=\tilde{\hv}_{k}^T\vv_{d_{k'}}$ for any $k'\neq k$. The interference term $z_{k}$ in \eqref{eq:receive} lies power $\E[|z_{k}|^2]\doteq 1$ due to the CSI error scaling \eqref{def:scaling}; and for the effective noise terms, we have $\E[|\hv_{k}^T \qv|^2]+\E[|n_k|^2] \doteq 1$.

It follows that the edge-encoded symbols $c_{d_k}$ can be decoded first by user $k$ at rate $R_E=(1-\alpha)\log P$ by treating the interference-plus-noise-term, of power exponentially equal to $P^{\alpha}$, as noise. Having canceled the signal $\tilde{x}_{E,k}$, the user $k$ can decode the cloud-encoded signal $s_{d_k}$ at rate $R_F=\alpha \log P$. We now analyze the resulting edge NDT. To this end, we define the time $T_{E1}=\mu L/((1-\alpha)\log P)$ required to decode reliably the edge-precoded signals and the time $T_{E2}=(1-\mu)L/(\alpha\log P)$ needed for the cloud-precoded signals to be reliably decoded.


When $\mu\leq(1-\alpha)$, we have $T_{E1}\leq T_{E2}$. In this case, due to the poor CSI at the cloud, the edge NDT is dominated by the latency $T_{E2}$ and, as a result, the total NDT is $\delta=(1-\mu)(1/\alpha+1/r)$. 


When $\mu\geq(1-\alpha)$ instead, we have $T_{E1}\geq T_{E2}$ and the delivery latency is dominated by the transmission of edge-precoded signals. After time $T_{E2}\leq T_{E1}$, the cloud-precoded signal are reliably decoded and hence we can set $\xv_F=0$ from the rest of the transmit duration $T_{E}-T_{E2}$. The delivery of a fraction $T_{E2}/T_{E1}$ of the edge-precoded signal is completed by time $T_{E2}$. The remaining fraction $(1-T_{E2}/T_{E1})$ of each cached subfile can be sent via ZF beamforming at the edge at rate $R'_E=\log P$ bits without interference, and the required time is given as $T'_{E1}=\mu L(1-T_{E2}/T_{E1})/\log P$, yielding the total edge time $T_E=T_{E2}+T'_{E1}=L/\log P$, i.e., $\delta_E=1$. Hence, the overall NDT for this case is $\delta=(1-\mu)/r+1$. This completes the proof.

\bibliographystyle{IEEEtran}
\bibliography{IEEEabrv,final_refs}

\end{document}